\documentclass[twoside,a4paper,11pt]{article}
\usepackage{amsmath}
\usepackage{amsthm}
\usepackage{amssymb}
\usepackage{amscd}
\usepackage{graphicx}
\usepackage{epsfig}
\usepackage{bm}

\usepackage{latexsym}
\usepackage{amsfonts}
\usepackage{color}
\usepackage{bbm}

\input xy
\xyoption{all}

-\headheight\advance\topmargin by -\headsep

\font\black=cmbx10 \font\sblack=cmbx7 \font\ssblack=cmbx5 \font\blackital=cmmib10  \skewchar\blackital='177
\font\sblackital=cmmib7 \skewchar\sblackital='177 \font\ssblackital=cmmib5 \skewchar\ssblackital='177
\font\sanss=cmss11 \font\ssanss=cmss8 
\font\sssanss=cmss8 scaled 600 \font\blackboard=msbm10 \font\sblackboard=msbm7 \font\ssblackboard=msbm5
\font\caligr=eusm10 \font\scaligr=eusm7 \font\sscaligr=eusm5  \font\fraktur=eufm10
\font\sfraktur=eufm7 \font\ssfraktur=eufm5

\font\bsymb=cmsy10 scaled\magstep2
\def\all#1{\setbox0=\hbox{\lower1.5pt\hbox{\bsymb
       \char"38}}\setbox1=\hbox{$_{#1}$} \box0\lower2pt\box1\;}
\def\exi#1{\setbox0=\hbox{\lower1.5pt\hbox{\bsymb \char"39}}
       \setbox1=\hbox{$_{#1}$} \box0\lower2pt\box1\;}

\def\tx#1{{\fam0\relax#1}}

\newfam\bifam
\textfont\bifam=\blackital \scriptfont\bifam=\sblackital \scriptscriptfont\bifam=\ssblackital

\newfam\blfam
\textfont\blfam=\black \scriptfont\blfam=\sblack \scriptscriptfont\blfam=\ssblack

\newfam\bbfam
\textfont\bbfam=\blackboard \scriptfont\bbfam=\sblackboard \scriptscriptfont\bbfam=\ssblackboard

\newfam\ssfam
\textfont\ssfam=\sanss \scriptfont\ssfam=\ssanss \scriptscriptfont\ssfam=\sssanss
\def\sss#1{{\fam\ssfam\relax#1}}

\newfam\clfam
\textfont\clfam=\caligr \scriptfont\clfam=\scaligr \scriptscriptfont\clfam=\sscaligr

\newfam\frfam
\textfont\frfam=\fraktur \scriptfont\frfam=\sfraktur \scriptscriptfont\frfam=\ssfraktur

\def\hpb#1{\setbox0=\hbox{${#1}$}
    \copy0 \kern-\wd0 \kern.2pt \box0}
\def\vpb#1{\setbox0=\hbox{${#1}$}
    \copy0 \kern-\wd0 \raise.08pt \box0}

\def\pmb#1{\setbox0\hbox{${#1}$} \copy0 \kern-\wd0 \kern.2pt \box0}
\def\pmbb#1{\setbox0\hbox{${#1}$} \copy0 \kern-\wd0
      \kern.2pt \copy0 \kern-\wd0 \kern.2pt \box0}
\def\pmbbb#1{\setbox0\hbox{${#1}$} \copy0 \kern-\wd0
      \kern.2pt \copy0 \kern-\wd0 \kern.2pt
    \copy0 \kern-\wd0 \kern.2pt \box0}
\def\pmxb#1{\setbox0\hbox{${#1}$} \copy0 \kern-\wd0
      \kern.2pt \copy0 \kern-\wd0 \kern.2pt
      \copy0 \kern-\wd0 \kern.2pt \copy0 \kern-\wd0 \kern.2pt \box0}
\def\pmxbb#1{\setbox0\hbox{${#1}$} \copy0 \kern-\wd0 \kern.2pt
      \copy0 \kern-\wd0 \kern.2pt
      \copy0 \kern-\wd0 \kern.2pt \copy0 \kern-\wd0 \kern.2pt
      \copy0 \kern-\wd0 \kern.2pt \box0}

\mathchardef\za="710B  
\mathchardef\zb="710C  
\mathchardef\zg="710D  
\mathchardef\zd="710E  
\mathchardef\zve="710F 
\mathchardef\zz="7110  
\mathchardef\zh="7111  
\mathchardef\zvy="7112 
\mathchardef\zi="7113  
\mathchardef\zk="7114  
\mathchardef\zl="7115  
\mathchardef\zm="7116  
\mathchardef\zn="7117  
\mathchardef\zx="7118  
\mathchardef\zp="7119  
\mathchardef\zr="711A  
\mathchardef\zs="711B  
\mathchardef\zt="711C  
\mathchardef\zu="711D  
\mathchardef\zvf="711E 
\mathchardef\zq="711F  
\mathchardef\zc="7120  
\mathchardef\zw="7121  
\mathchardef\ze="7122  
\mathchardef\zy="7123  
\mathchardef\zf="7124  
\mathchardef\zvr="7125 
\mathchardef\zvs="7126 
\mathchardef\zf="7127  
\mathchardef\zG="7000  
\mathchardef\zD="7001  
\mathchardef\zY="7002  
\mathchardef\zL="7003  
\mathchardef\zX="7004  
\mathchardef\zP="7005  
\mathchardef\zS="7006  
\mathchardef\zU="7007  
\mathchardef\zF="7008  
\mathchardef\zW="700A  

\newcommand{\be}{\begin{equation}}
\newcommand{\ee}{\end{equation}}

\newcommand{\bea}{\begin{eqnarray}}
\newcommand{\eea}{\end{eqnarray}}
\newcommand{\beas}{\begin{eqnarray*}}
\newcommand{\eeas}{\end{eqnarray*}}
\def\*{{\textstyle *}}
\newcommand{\R}{{\mathbb R}}

\newcommand{\C}{{\mathbb C}}

\newcommand{\ot}{\otimes}

\newcommand{\rmp}{\textnormal{p}}
\newcommand{\pa}{\partial}
\newcommand{\ti}{\times}

\newcommand{\cG}{{\mathcal G}}

\def\cH{{\cal H}}

\def\Sec{\operatorname{Sec}}
\def\Ad{\operatorname{Ad}}

\def\bY{{\mathbf{Y}}}
\def\bX{{\mathbf{X}}}

\def\bZ{{\mathbf{Z}}}

\def\sT{{\sss T}}

\def\sv{{\sss v}}

\def\xd{\tx{d}}
\def\xi{\tx{i}}

\newdir{|>}{%
!/4.5pt/@{|}*:(1,-.2)@^{>}*:(1,+.2)@_{>}}

\def\SO{\operatorname{SO}}
\def\so{\operatorname{so}}

\newdir{ (}{{}*!/-5pt/@^{(}}

\newcommand{\tr}{\mbox{$\mathrm{tr}$}}

\newcommand{\ket}[1]{\ensuremath{| #1\rangle}}
\newcommand{\bk}[2]{\ensuremath{\langle #1 | #2\rangle}}
\newcommand{\kb}[2]{\ensuremath{| #1\rangle\!\langle #2 |}}

\newcommand{\re}{{\mathcal{R}e}}
\newcommand{\id}{\mathbbmss{1}}
\newcommand{\rmb}{\textnormal{b}}

\newtheorem{theorem}{Theorem}[section]
\newtheorem{proposition}[theorem]{Proposition}
\newtheorem{lemma}[theorem]{Lemma}
\newtheorem{corollary}[theorem]{Corollary}

\theoremstyle{definition}
\newtheorem{example}[theorem]{Example}
\newtheorem{remark}[theorem]{Remark}

\begin{document}

\title{Lie groupoids in information geometry}
\author{Katarzyna Grabowska\footnote{email:konieczn@fuw.edu.pl }\\
\textit{Faculty of Physics,
                University of Warsaw}
\\
\\
Janusz Grabowski\footnote{email: jagrab@impan.pl} \\
\textit{Institute of Mathematics, Polish Academy of Sciences} \\
\\
Marek Ku\'s\footnote{email: marek.kus@cft.edu.pl}\\
\textit{Center for Theoretical Physics, Polish Academy of Sciences}
 \\
\\
Giuseppe Marmo\footnote{email: marmo@na.infn.it}\\
\textit{Dipartimento di Fisica ``Ettore Pancini'', Universit\`{a} ``Federico II'' di Napoli} \\
\textit{and Istituto Nazionale di Fisica Nucleare, Sezione di Napoli} \\ \\
}
\date{}
\maketitle
\begin{abstract}
\textit{We demonstrate that the proper general setting for contrast (potential) functions in statistical and information geometry is the one provided by Lie groupoids and Lie algebroids. The contrast functions are defined on Lie groupoids and give rise to two-forms and three-forms on the corresponding Lie algebroid. If the two-form is non-degenerate, it defines a `pseudo-Riemannian' metric on the Lie algebroid and a family of Lie algebroid torsion-free connections, including the Levi-Civita connection of the metric. In this framework, the two-point functions
are just functions on the pair groupoid $M\ti M$ with the `standard' metric and affine connection on the Lie algebroid $\sT M$. We study also reductions of such systems and infinite-dimensional examples. In particular, we find a contrast function defining the Fubini-Study metric on the Hilbert projective space.}
\end{abstract}

\section{Introduction}
Information geometry studies statistical-probabilistic models equipped with a differential structure. It started with the pioneering work of Rao \cite{Rao1945} and was brought to a mature level with works of Amari
\cite{Amari2007, Amari2016,Amari2012} and many others \cite{Barndorff1987,Barndorff1997,Blesild1991}.

The main consideration is that when dealing with probability distributions, two aspects play a relevant role. The first one is a notion of ``distinguishability'' or ``distance'' measuring the relative difference between two probability distributions, the second one is connected with the possibility to compose distributions by means of convex combinations, i.e., treating them as elements of a convex set of an affine space. Information geometry aims at treating these aspects from the point of view of differential geometry, where the notion of distance will be associated with a metric by means of geodesic distances, while the composition will be associated with a connection. In general the connection will not be the Levi-Civita one and, in many instances, it turns out to have curvature different form zero; therefore one looks for alternative connections, which are (possibly) flat and would allow for a notion of a ``convex'' composition.
This is equivalent to the existence of an affine coordinate system.

Our aim in this note is to convince the reader that the proper general setting for contrast (potential) functions in statistical and information geometry is the one provided by Lie groupoids and Lie algebroids. The contrast functions are defined on Lie groupoids (in the standard case the pair groupoid $M\ti M$) and vanish on the submanifold of units (the diagonal in the standard case). The corresponding statistical manifolds are given by a two-form and a three-form on the corresponding Lie algebroid. If the two-form is non-degenerate, it defines a `pseudo-Riemannian' metric on the Lie algebroid and a family of Lie algebroid torsion-free connections, including the Levi-Civita connection of the metric. The three-tensor controls the deviation for the Levi-Civita connection. In this framework, the two-point functions are just functions on the pair groupoid $M\ti M$ with the `standard' metric and affine connection on the Lie algebroid $\sT M$. Naturally understood reductions can lead from the `two-point function case' to contrast function on more complicated manifolds (see Example \ref{e1}).

This generalization is not only completely natural, but offers a wide field of new applications and examples.
Our approach is coordinate-free, so that the corresponding differential calculus can be done in the framework of Banach manifolds. In particular, we study a contrast function defining the Fubini-Study metric on the Hilbert projective space in an infinite dimension.

\section{Information geometry}
As it was written above, the main underlying idea of \emph{information geometry} is to give and analyze a geometric structure of the set of probability distributions pertinent to the problem in question \cite{Amari2007,Amari2016,Ciaglia2017}. To this end one introduces a \emph{statistical manifold}, i.e., a triple $(\mathcal{M}, g, T)$, where $\mathcal{M}$ is a differential manifold parameterizing a family of probability distribution, $g$ is a metric tensor on $\mathcal{M}$, and $T$ is a third order \emph{skewness tensor} on $\mathcal{M}$ characterizing its \emph{flatness}.

Let us first start with the classical context. Let $\mathcal{P}(\mathcal{X})$ denote the space of probability distributions on a measure space $\mathcal{X}$ with a measure $d\mathbf{x}$. The statistical manifold $\mathcal{M}$ gives a parameterization of (a submanifold of) $\mathcal{P}(\mathcal{X})$ by an injective map $\mathcal{M}\ni m\mapsto p(\mathbf{x},m)d\mathbf{x}\in \mathcal{P}(\mathcal{X})$. In a coordinate system $\left\lbrace \zeta^j\right\rbrace $ the Fisher-Rao metric (the tensors $g$) and the skewness tensor $T$ take the form
\begin{equation}\label{eq:g}
g_{jk}  = \int_\mathcal{X}
p(\mathbf{x},\mathbf{\zeta})
\left( {\frac{{\partial \log \left(p(\mathbf{x},\mathbf{\zeta}) \right)}}{{\partial \zeta ^j }}}\right)
\left( {\frac{{\partial \log \left(p(\mathbf{x},\mathbf{\zeta}) \right)}}{{\partial \zeta ^k }}} \right)d\mathbf{x}
\end{equation}
\begin{equation}\label{eq:T}
T_{jkl}  = \int_\mathcal{X}
p(\mathbf{x},\mathbf{\zeta})
\left( {\frac{{\partial \log \left(p(\mathbf{x},\mathbf{\zeta}) \right)}}{{\partial \zeta ^j }}}\right)
\left( {\frac{{\partial \log \left(p(\mathbf{x},\mathbf{\zeta}) \right)}}{{\partial \zeta ^k }}} \right)
\left( {\frac{{\partial \log \left(p(\mathbf{x},\mathbf{\zeta}) \right)}}{{\partial \zeta ^k }}} \right)
d\mathbf{x}
\end{equation}
For a given tensors $g$ and $T$ we define a family of torsionless connections $\nabla^\alpha$ by its Christoffel symbols
\begin{equation}\label{eq:nablaalpha}
\Gamma^\alpha_{jkl}:=\Gamma^{\mathrm{LC}}_{jkl}-\frac{\alpha}{2}T_{jkl},
\end{equation}
where $\Gamma^\mathrm{LC}_{jkl}$ are the Christofell symbols for the metric tensor $g$ (the Levi-Civita connection).
	
For all vector fields $X,Y,Z$ on $\mathcal{M}$ we have a duality property
\begin{equation}\label{eq:dual}
Z\left(g(X,Y)\right)= g\left(\nabla^\alpha_ZX,Y\right) +g\left(X, \nabla^{-\alpha}_Z Y \right).
\end{equation}	
A torsionless connection is self-dual if $\nabla^\alpha=\nabla^{-\alpha}$, what implies $T=0$ and, consequently, identifies the Levi-Civita connection as the only torsionless one that is self-dual.

An alternative characterization of the geometric properties of the statistical structure on $\mathcal{M}$ can be given by introducing two torsionless connections $\nabla=\nabla^1$ and $\nabla^\ast=\nabla^{-1}$ in terms of which $T_{jkl}=\Gamma^\ast_{jkl}-\Gamma_{jkl}$. If both $\nabla$ and $\nabla^\ast$ are flat the statistical manifold $(\mathcal{M}, g, T)= (\mathcal{M}, g, \nabla, \nabla^\ast)$ is called \emph{dually flat} \cite{Amari2007,Amari2016,Ay2002}. It is to note that the space of pure states of a finite-level quantum system treated as a statistical manifold does not admit a dually flat structure \cite{Ay2002, Ay2003}.

The above outlined geometrical structure of a statistical manifold can be generalized in the following different way \cite{Amari2007, Amari2016}.  We introduce a two-point \emph{potential function} $F:\mathcal{M}\times \mathcal{M}\rightarrow \mathbb{R}$ which is usually a ``directed'' distance quantifying the relative distinguishability of two probability distributions \cite{Ciaglia2017} and thus often goes under the name of \emph{contrast function} or \emph{divergence}. For all $m_1,m_2\in\mathcal{M}$ the potential function is non-negative, $F(m_1,m_2)\ge 0$, and vanishes exactly on the diagonal, i.e.,  $F(m_1,m_2)=0$, if and only if $m_1=m_2$. Let, as above, $\left\lbrace \zeta^j\right\rbrace $ be a coordinate system on the first manifold $\mathcal{M}$ and $\left\lbrace \zx^j\right\rbrace $ on the second. If $F$ is at least $C^3$ the condition imposed on $F$ imply \cite{Matumoto1993},
\begin{equation}\label{eq:Ds}
\left. \frac{\partial F}{\partial \zeta^j}\right|_{\zeta=\zx}=\left. \frac{\partial F}{\partial \zx^j}\right|_{\zeta=\zx}=0.
\end{equation}
The metric and the torsion tensors are then given as
\begin{equation}\label{eq:gcoord}
g_{jk}=\left. \frac{\partial^2F}{\partial \zeta^j\partial \zeta^k}\right|_{\zeta=\zx}=\left. \frac{\partial^2F}{\partial \zx^j\partial \zx^k}\right|_{\zeta=\zx}=-\left. \frac{\partial^2F}{\partial \zx^j\partial \zeta^k}\right|_{\zeta=\zx},
\end{equation}
and
\begin{equation}\label{key}
T_{jkl}=\left. \frac{\partial^3F}{\partial \zeta^l\partial \zx^k\partial \zx^lj}\right|_{\zeta=\zx}=-\left. \frac{\partial^3F}{\partial \zx^l\partial \zeta^k\partial \zeta^j}\right|_{\zeta=\zx}.
\end{equation}
The condition (\ref{eq:Ds}) says that if we immerse $ M $ diagonally, $ \imath :M\rightarrow M\times M $, then $ \imath^\ast(dF)=0 $, and, taking into account the non-negativity of $ F $, the imposed conditions mean thus that it has a local minimum on the diagonal.

If additional requirements are imposed on $F$, they will provide the metric with additional properties. Thanks to the geometrical formulation of quantum mechanics \cite{Ashtekar1999}, it is possible to use this description also in the quantum setting.

\section{Lie groupoids}\label{Lie_groupoids}
To make the paper relatively self-contained we decided to summarize in two next sessions basic informations about Lie groupoids and Lie algebroids. In our presentation we used, only slightly adapted, lecture notes by Meinrenken \cite{Meinrenken2017}, but one can also use the book of Mackenzie \cite{Mackenzie2005} as a source of concepts, examples and references.

The structure of a \emph{Lie groupoid}, $\mathcal{G}\rightrightarrows M$, involves a manifold $ \mathcal{G} $ of arrows, a submanifold $ \imath: M \hookrightarrow\mathcal{G} $ of units (objects), and two surjective submersions $\mathsf{s},\mathsf{t}:\mathcal{G}\rightarrow M $, called \emph{source} and \emph{target}, such that $ \mathsf{t}\circ\imath= \mathsf{s}\circ\imath=\mathrm{id}_M$.

One thinks of $ g\in\cG$ as of an arrow from its source $ \mathsf{s}(g) $ to its target $ \mathsf{t}(g) $, with $M$ embedded as trivial arrows. The arrows $ g_1 $ and $ g_2 $ can be composed, $ g_1 \circ g_2$, provided $ \mathsf{s}(g_1) = \mathsf{t}(g_2) $:
\begin{equation}\label{eq:composition}
\circ:\mathcal{G}^{(2)}:=\left\lbrace (g_1,g_2)\in\mathcal{G}^2, \mathsf{s}(g_1) = \mathsf{t}(g_2)\right\rbrace \ni (g_1,g_2)\mapsto g_1 \circ g_2\in\mathcal{G},
\end{equation}
such that $ \mathsf{t}(g_1 \circ g_2)=\mathsf{t}(g_1) $, $ \mathsf{s}(g_1 \circ g_2)=\mathsf{s}(g_1) $. The composition $ \circ $ is associative, i.e., $ (g_1 \circ g_2)\circ g_3=g_1 \circ (g_2\circ g_3)$, whenever
$$ (g_1,g_2,g_3) \in \mathcal{G}^{(3)}:=\left\lbrace (g_1,g_2,g_3)\in\mathcal{G}^3: \mathsf{s}(g_1) = \mathsf{t}(g_2), \mathsf{s}(g_2) = \mathsf{t}(g_3)\right\rbrace\,.$$

Elements $x\in M $ act as units $\id_{x}$, $ \mathsf{t}(g) \circ g=g=g\circ\mathsf{s}(g) $, and there is an inverse map $ \mathrm{inv}:\mathcal{G}\ni g\mapsto g^{-1}\in \mathcal{G} $, such that
$ \mathsf{s}(g^{-1}) = \mathsf{t}(g) $, $ \mathsf{t}(g^{-1}) = \mathsf{s}(g)$, and $ g\circ g^{-1} $ and $ g^{-1}\circ g $ are units, $\id_{\mathsf{t}(g)}$ and $\id_{\mathsf{s}(g)}$, respectively.

Note that the groupoid structure of $ \mathcal{G} $ is completely determined by the graph of the composition,
\begin{equation}\label{eq:graph}
\mathrm{Gr}(\mathcal{G})=\left\lbrace (g_1,g_2,g_3)\in\mathcal{G}^3: g_3=g_1 \circ g_2\right\rbrace.
\end{equation}
Whenever we write $ g_1 \circ g_2 $, we implicitly assume that $ g_1 $ and $ g_2 $ are composable.

A \emph{morphism of Lie groupoids}, $\Phi:\mathcal{G}_1\rightarrow\mathcal{G}_2 $ is a smooth map, such  that $ \Phi(g_1 \circ g_2 )=\Phi(g_1)\circ \Phi(g_2)$.
\begin{example}
	A Lie group is a Lie groupoid with a unique unit.	
\end{example}
\begin{example}
	For any manifold $M$ one can construct the \emph{pair groupoid} $ M\times M\rightrightarrows M $, with $ \mathsf{s}(m,m^\prime)=m^\prime $, $ \mathsf{t}(m,m^\prime)=m $ and
$$ (m, m^\prime) = (m_1,m_1^\prime)\circ (m_2,m_2^\prime) \Leftrightarrow m_1^\prime=m_2\,, m=m_1\,, m_2^\prime=m^\prime\,.$$
The units are given by the diagonal embedding $ D:M \hookrightarrow M\times M $.	
\end{example}
\begin{example}
	Let $ \kappa: P\rightarrow M $ be a $G$-principal bundle. The \emph{Atiyah groupoid} $ \mathcal{G}(P) $ is the groupoid structure on $ P\times P/G=\left\lbrace \left[ \left(x,y\right) \right]:x,y \in P \right\rbrace\rightrightarrows M  $, where the class $ \left[ \left(x,y\right) \right] $ is taken with respect to the equivalence relation $ (x,y)\sim(x^\prime, y^\prime)\Leftrightarrow\exists g\in G: x^\prime =gx, y^\prime= gy $. We have
	$ \mathsf{s}\left(\left[ \left( x,y \right)   \right] \right)=\kappa(y)  $, $ \mathsf{t}\left(\left[ \left( x,y \right)   \right] \right)=\kappa(x)  $, and $ \left[ \left( x,y \right)   \right]\circ \left[ \left( x^\prime,y^\prime \right)   \right]=\left[ \left( x,y^\prime g \right)   \right] $, provided $x^\prime g=y$. Of course, $ \left[ \left( x,y \right)   \right]^{-1}=\left[ \left( y,x \right)   \right] $.
 \end{example}

\section{Lie algebroids}\label{Lie_algebroids}
Infinitesimal parts of Lie groupoids are \emph{Lie algebroids}. A Lie algebroid over $M$ is a vector bundle $ \tau: E\rightarrow M$ together with a Lie bracket $\left[.,. \right] $ on the space of its sections, such that there exists a vector bundle map $ {\za}:E\rightarrow\mathsf{T}M $ covering the identity on $M$, called the \emph{anchor map}, satisfying the Leibniz rule,
$$ \left[X, fY \right]=f\left[X,Y \right]+({\za}(X) f) Y $$ for all $X,Y\in\Sec(E)$,  $ f\in C^\infty(M)$.
\begin{example}
	A Lie algebroid over a point is the same as a finite-dimensional real Lie algebra.
\end{example}
\begin{example}
	The tangent bundle $ \mathsf{T}M $ with its usual bracket of vector fields is a Lie algebroid wit $ {\za}=\mathrm{id }_\mathsf{T}M  $.	
\end{example}
\begin{example}
	For every principal $G$-bundle, $ \kappa: P\rightarrow M  $ the bundle $ E(p)=\mathsf{T}P/G\rightarrow P/G=M $ is a Lie algebroid called the \emph{Atiyah algebroid} of $P$. The bracket on $ \Sec(E) $ is induced from the identification of elements of $\Sec(E)$ with $G$-invariant vector fields on $P$, while the anchor map is induced from $ \mathsf{T}\kappa:\mathsf{T}P\rightarrow TM $. 	
\end{example}
\subsection{The Lie algebroid of a Lie groupoid}\label{La_of_Lg}
We now define the Lie algebroid $ E=\mathrm{Lie}(\mathcal{G}) $ of a
Lie groupoid $ \mathcal{G}\rightrightarrows M $. As a vector bundle,
we take $\mathrm{Lie}(\mathcal{G})=\nu(\mathcal{G},M)=\mathsf{T}\mathcal{G}|_M/\mathsf{T}M $ to be the normal bundle of $M$ in $ \mathcal{G} $. To define the anchor map, note that $ \mathsf{s} $ and  $ \mathsf{t} $ coincide on $ M\subset \mathcal{G} $, so the difference $ (\mathsf{T}\mathsf{s} -\mathsf{T}\mathsf{t}):\mathsf{T}\mathcal{G}\rightarrow \mathsf{T}M$ vanishes on $ \mathsf{T}M\subset \mathsf{T}\mathcal{G}$ and hence descend to a map $ {\za}:\nu (\mathcal{G},M)\rightarrow \mathsf{T}M$.

A vector field $ \tilde{X} $ on $ \mathcal{G} $ is called \emph{left-invariant} if it is tangent to target fibers and $ \mathsf{T}L_g \tilde{X}_h=\tilde{X}_{g\circ h}$ (note the left multiplication). Similarly $ \tilde{X} $ is \emph{right-invariant} if it is tangent to source fibers and satisfies $ \mathsf{T}R_g \tilde{X}_h=\tilde{X}_{h\circ g}$. By construction, the spaces $\mathfrak{X}^L(\mathcal{G})$ and $\mathfrak{X}^R(\mathcal{G})$ of left (resp., right) invariant fields on $ \mathcal{G} $ form Lie subalgebras. Since $ \mathrm{Ker}\mathsf{T}\mathsf{t}|_M $, $ \mathrm{Ker}\mathsf{T}\mathsf{s}|_M $ are complements to $ \mathsf{T}M $ in $ \mathsf{T}\mathcal{G}|_M $, each of this bundles may be identified with the normal bundle.

For $X\in\Sec(\mathrm{Lie}(\mathcal{G}))$ we denote by $ X^L\in\mathfrak{X}_L(\mathcal{G}) $ and $ X^R\in\mathfrak{X}_R(\mathcal{G}) $ the unique left-invariant and right-invariant vector fields, such that $ X^L|_M\sim X $ and  $ X^R|_M\sim X $.
\begin{proposition}
	For all $X\in\Sec(\mathrm{Lie}(\mathcal{G}))$ we have
\be X^L\sim_\mathsf{t}0\,, X^L\sim_\mathsf{s}{\za}(X)\,, X^R\sim_\mathsf{t}-{\za}(X)\,, X^R\sim_\mathsf{s}0\,\ \text{and}\,\ {\za(X)}\sim_\imath X^L-X^R\,,\ee
where $\imath: M\rightarrow\mathcal{G}$ is the inclusion of units. Furthermore
\be\label{a} X^L\sim_{\mathrm{inv}_\mathcal{G}}-X^R\,.
\ee
Here, $ \tilde{X}\sim_f \tilde{Y} $ means that the vector fields $ \tilde{X} $ and $ \tilde{Y}$ are $f$-related. Moreover
\be (fX)^L=f^LX^L\,\ \text{and}\,\  (fX)^R=f^RX^R\,,\ee
where $f^L=f\circ \mathsf{s}$ and $f^R=f\circ \mathsf{t}$.
\end{proposition}
\begin{proposition}
	There exists a unique Lie bracket $ \left[.,. \right]  $ on $ \Sec\left(\mathrm{ Lie}(\mathcal{G})\right)  $, such that
\be\left[X^L,Y^L \right]=\left[X,Y\right]^L\,\  \left[X^R,Y^L \right]=0\,\ \left[X^R,Y^R \right]=-\left[X,Y\right]^R\,.\ee
This is a Lie algebroid bracket with the anchor ${\za}$. Moreover,
\be X^L(f^L)=({\za}(X)f)^L\,,\  X^R(f^L)=0\,,\ X^L(f^R)=0\,,\  X^R(f^R)=-({\za}(X)f)^R\,.\ee
\end{proposition}
\begin{example}
	If $ \mathcal{G}=G $ is a Lie group, then $ \mathrm{Lie}(\mathcal{G}) $ is the Lie algebra of $G$.
\end{example}
\begin{example}
	If $ \mathcal{G}=M\times M $ is a pair groupoid, then $ \mathrm{Lie}(\mathcal{G})=\mathsf{T}M $.
\end{example}
\begin{example}
	If $ \mathcal{G}=P\times P/G $ is an Atiyah groupoid associated with a $G$-principal bundle $P$, then $ \mathrm{Lie}(\mathcal{G})=\mathsf{T}P/G $ is the Atiyah algebroid of $P$.
\end{example}

\subsection{Lie algebroid connections}
Let $E\rightarrow M$ be a Lie algebroid. An $E$-\emph{connection} on  a real vector bundle $V\rightarrow M$ is a bilinear map
$$ \nabla:\Sec(E)\times \Sec(V)\rightarrow \Sec(V)\,,\ (X,\sigma)\mapsto \nabla_X \sigma\,,$$
with the properties: $\nabla_{fX} = f\nabla_X $, $ \nabla_X (f\sigma) = f\nabla_X \sigma +{\za}(X)(f)\sigma$.

For $E=\mathsf{T}M$ the $E$-connection is the standard affine connection. On the other hand, every $\mathsf{T}M$-connection $\tilde{\nabla}$ determines an $E$-connection by setting $\nabla_X=\tilde{\nabla}_{{\za}(X)}$.

The curvature of an $E$-connection is tensor field $ \mathrm{Curv}^\nabla \in \Sec(\Lambda^2 E^\ast\ot\mathrm{End}(V))$ defined by
$$ \mathrm{Curv}^\nabla(X,Y) = \left[\nabla_X,\nabla_Y\right]-\nabla_{\left[X,Y \right]}=\nabla_X \nabla_Y - \nabla_Y \nabla_X- \nabla_{\left[X,Y \right]}\,.$$
Thus, $ \mathrm{Curv}^\nabla=0 $ if and only if the map $X\mapsto \nabla_X$ preserves the brackets. In this case the connection is called \emph{flat} or a \emph{representation} of $E$.

In the case $V=E$, the \emph{torsion} of an $E$-connection is the tensor field $\mathrm{Tor}^\nabla\in\Sec(\Lambda^2E\ot E)$ given by
$$\mathrm{Tor}^\nabla(X,Y)=\nabla_X Y-\nabla_Y X-\left[X,Y \right]\,.$$

A \emph{pseudo-Riemannian metric} on $E$ is a tensor field $g\in\Sec(\bigodot^2E^\ast)$, where $ \bigodot^2E^\ast $ is the symmetric tensor product $ E^\ast\odot E^\ast $, such that $\Sec(E)\ni X \mapsto g(X,\cdot)\in\Sec(E^\ast)$ defines isomorphism of vector bundles.

\begin{proposition}
	For every pseudo-Riemannian metric $g$ on $E$ there is a unique torsion-free $E$-connection $\nabla^g$ on $E$ that is \emph{metric}, i.e. $g\left(\nabla^g_X Y,Z \right)+g\left(Y,\nabla^g_X Z\right)={\za}(X)g(Y,Z)$.
	We call it the \emph{Levi-Civita connection} of $g$.
\end{proposition}
The proof is the same as in the standard case:
$$ 2g(\nabla^g_X Y,Z) ={\za}(X)g(Y,Z)+{\za}(Y)g(Z,X)-{\za}(Z)g(X,Y)+g(\left[X,Y\right],Z )-g(\left[Y,Z\right],X )-g(\left[X,Z\right],Y)\,.$$

The \emph{dual connection} of an $E$-connection $\nabla$ on a pseudo-Riemannian Lie algebroid $(E,g)$ is the connection $\nabla^\ast$ defined by
$$\za(X)g(Y,Z)=g(\nabla_X Y,Z)+g(Y,\nabla_X^\ast Z)\,.$$
Hence, the Levi-Civita connection is a self-dual, torsion-free connection.

\section{Contrast functions on Lie groupoids and dualistic structures}
Our aim in this section is to show that the right framework for contrast functions, as they were defined for stochastic models above, is the theory of the Lie groupoids and Lie algebroids. In this sense, the two point contrast function $F=F(x,y)$ is viewed as a one-point function of the pair groupoid $\mathcal{G}=M\times M$. The metric induced by $F$ in the standard case is the (pseudo-Riemannian) metric on the Lie algebroid $\mathrm{Lie}(\mathcal{G})$. In the standard case $\cG=M\ti M$ it reduces to a metric on $\mathsf{T}M$ or, as we use to say in such a case, on $M$. Also the pair of dual connections defined by the two-point contrast function $F(x,y)$ can be generated as the pair of dual $E$-connections with $E=\mathrm{Lie}({\mathcal{G})}$. In this sense, the
expected and observed $\za$-geometries of a statistical model introduced by Chentsov
and Amari (cf. \cite{Amari1985})
are particular instances of geometries derived from contrast functions on Lie groupoids. Consequently,
these statistical geometries may be studied within this unified framework.

We start with the following Lemma,
\begin{lemma}\label{L1}
	Let $M$ be a submanifold in $N$ and let $F:N\rightarrow\mathbb{R}$ vanishes on $M$ together with all its derivatives up to an order $k$, i.e. the $k$-th jet of $F$ vanishes on $M$, $\mathsf{j}^kF|_M=0$. Then, for vector fields $\tilde{X}_1,\ldots,\tilde{X}_{k+1}$ on $N$, the derivative
$$\omega(\tilde{X}_1,\ldots,\tilde{X}_{k+1})=\tilde{X_1}\cdots\tilde{X}_{k+1}(F)|_M$$ defines a symmetric $(k+1)$-tensor on $M$ that depends only on the class $ X_1,\ldots,X_{k+1}\in \nu(N,M)$ in the normal bundle $ \nu(N,M)=\mathsf{T}N|_M/\mathsf{T}M$. In other words, $\omega\in\bigodot^{k+1}\nu^*(N,M)$.
\end{lemma}
\begin{proof}
	Observe first that $\omega$ is symmetric. For it suffices to show that it does not change under the transposition of neighboring arguments. We have
	\begin{eqnarray*}
	 \tilde{X_1}\cdots\tilde{X}_{k+1}(F)\_M&-&\tilde{X_1}\cdots\tilde{X}_{i-1}\tilde{X}_{i+1}\tilde{X}_{i}\tilde{X}_{i+2}\cdots \tilde{X}_{k+1}(F)|_M= \\
	&=&\tilde{X_1}\cdots\tilde{X}_{i-1}\left[ \tilde{X}_{i+1},\tilde{X}_{i}\right] \tilde{X}_{i+2}\cdots \tilde{X}_{k+1}(F)|_M=0,
	\end{eqnarray*}
	since the $k$-th jet of $F$ vanishes on $M$. Further, $\omega$ is a tensor, since $$\omega(f\tilde{X}_1,\ldots,\tilde{X}_{k+1})=(f\tilde{X_1})\cdots\tilde{X}_{k+1}(F)|_M=f|_M\,\omega(\tilde{X}_1,\ldots,\tilde{X}_{k+1})\,.$$ Finally, if $X_1=X_1^\prime$, then $\tilde{X}_1-\tilde{X}_1^\prime|_M\in\mathsf{T}M$, so that $(\tilde{X}_1-\tilde{X}_1^\prime)\tilde{X}_2\cdots\tilde{X}_{k+1} F|_M=0$ since   $\tilde{X}_2\cdots\tilde{X}_{k+1}F$ vanishes on $M$.
\end{proof}	

Let now $\mathcal{G}\rightrightarrows M$ be a Lie groupoid and $F:\mathcal{G}\rightarrow\mathbb{R}$ be a smooth function vanishing on $M\hookrightarrow \mathcal{G}$. We say that $F$ is a \emph{contrast function} if $dF|_M=0$. A direct consequence of the Lemma~\ref{L1} is the following
\begin{proposition}
	Any contrast function $F:\mathcal{G}\rightarrow\mathbb{R}$ defines on the Lie algebroid $E=\mathrm{ Lie}(\mathcal{G})$ a symmetric 2-form $g^F\in\bigodot^2E^\ast$ by $g^F(X,Y):=\tilde{X}\tilde{Y} F|_M$, where $ \tilde{X} $ and $ \tilde{Y} $ are any vector fields on $\mathcal{G}$ representing $X, Y \in \nu(\mathcal{G},M)$ at points of $M$. In particular,
$$ g^F(X,Y)=X^LY^LF|_M=X^LY^RF|_M=X^RY^RF|_M\,.$$
Of course, if $F\ge 0$, then $g^F$ is non-negatively defined.
\end{proposition}
A symmetric 2-form on $E$ we will call a \emph{pseudometric}. We call the contrast function $F$ \emph{regular} if $g^F$ has constant rank as a morphism of vector bundles $g^F:E\to E^*$, and \emph{metric} if $g^F$ is a pseudo-Riemannian (non-degenerate) metric on $E$, i.e. $g^F:E\to E^*$ is an isomorphism. If $F\ge 0$ is metric, then $g^F$ is a Riemannian metric on $E$.
\begin{example}\label{ex1}
Define $F:GL(n, \mathbb{R})\rightarrow\mathbb{R}$ by
$$F(A)=\tr(I-A)(I-A)^t\,.$$
Let $ X,Y\in\mathrm{Lie}(GL(n, \mathbb{R}))=gl(n, \mathbb{R}) $. We have
\beas X^LY^LF(I)&=&X^L\frac{d}{ds}_{|{s=0}}\,\tr(I-A\exp(sY))(I-A\exp(sY))^t\\
&=&X^L\left(-\tr(AY(I-A)^t+(I-A)Y^tA^t\right)(I)\\
&=&\frac{d}{ds}_{|{s=0}}\, \left(-\tr(\exp(sX)Y(I-\exp(sX))^t+(I-\exp(sX))Y^t\exp(sX^t)^t\right)\\
&=&\tr(YX^t+XY^t)=2\tr(XY^t)
\eeas
\end{example}

\begin{remark}
Metric contrast functions on $M\ti M$ are called \emph{yokes} in \cite{Barndorff1987} (see also \cite{Barndorff1997,Blesild1991} for the idea of generating tensors from yokes).
\end{remark}
\begin{proposition}
Let $F$ be a contrast function on $\mathcal{G}$. Then, $F^*=F\circ\mathrm{inv}$ is also a contrast function and $g^F=g^{F^*}$.
\end{proposition}
\begin{proof}
For $X,Y$ being sections of $\mathrm{ Lie}(\mathcal{G})$ we have (cf. (\ref{a}))
$$X^L(F^*)|M=X^L(F\circ\mathrm{inv})|M=-X^R(F)\circ\mathrm{inv}|M=0\,,$$
so $F^*$ is a contrast function. Moreover,
$$g^F(X,Y)=X^LY^L(F)|M=X^LY^L(F)\circ\mathrm{inv}|M=X^RY^R(F\circ\mathrm{inv})|M=g^{F^*}(X,Y)\,.
$$
\end{proof}

For a metric contrast function $F$ on $\mathcal{G}$ we define on $E=\mathrm{ Lie}(\mathcal{G})$ the Lie($ \mathcal{G} $)-connections
$\nabla_X^F $ and  $\nabla_X^{F^\ast} $ by
\be\label{b} g(\nabla_X^F Y,Z)=X^LY^LZ^RF|_M \,\ \text{and}\,\  g(\nabla_X^{F^\ast} Y,Z)=X^LY^LZ^RF^*|_M\,, \ee
where $ F^\ast=F\circ\mathrm{inv} $.
\begin{theorem}
	$ \nabla^F $ and $ \nabla^{F^\ast} $ are dual, torsion-free $E$-connections on $E$, such that $\frac{1}{2}\left(\nabla^F+\nabla^{F^\ast} \right)$ is the Levi-Civita connection $\nabla^g$ with respect to $g=g^F=g^{F^*}$. In particular, if $F=F^\ast$, the connection $ \nabla^F =\nabla^{F^\ast} $ is the Levi-Civita connection $\nabla^g$ for $g$.
\end{theorem}
\begin{proof}
	Let us first prove that $ \nabla^F $ and $ \nabla^{F^\ast} $ are $E$-connections. Linearity with respect to $X,Y,Z$ is clear. Since, for $f\in C^\infty(M)$ we have
	$$X^LY^L(fZ)^RF|_M=X^LY^Lf^RZ^RF|_M=fX^LY^LZ^RF|_M\,,$$
	the equations (\ref{b}) properly define $\nabla_X^FY$ and $\nabla_X^{F^*}Y$ as sections of $E$.
Similarly, $$g(\nabla_{fX}^FY,Z)=(fX)^LY^LZ^RF|_M=f^L|_MX^LY^LZ^RF|_M=fg(\nabla_X^FY,Z)\,.$$
Finally, \beas g(\nabla_X^FfY,Z)&=&X^L(fY)^LZ^RF|_M=X^Lf^LY^LZ^RF|_M=({\za}(X)f)^LY^LZ^RF|_M+f^LX^LY^LZ^RF|_M\\
&=&{\za}(X)(f)g(Y,Z)+fg(\nabla_X^FY,Z)\,,\eeas
so that $\nabla_X^FfY={\za}(X)(f)Y+f\nabla_X^FY$.
Now, we have
$$
g(\nabla_X^FY,Z)+g(Y,\nabla_X^{F^\ast}Z)=X^LY^LZ^RF|_M+X^LZ^LY^RF^\ast|_M\,.$$ But,
$$X^LZ^LY^RF^\ast=X^LZ^LY^R(F\circ\mathrm{inv})=-X^LZ^L(Y^L(F)\circ\mathrm{inv}) =-X^RZ^RY^L(F)\circ\mathrm{inv}\,.$$
Finally,
$$-X^RZ^RY^L(F)\circ\mathrm{inv}|_M =-X^RZ^RY^L(F)|_M\,,$$ so that
\beas g(\nabla_X^FY,Z)+g(Y,\nabla_X^{F^\ast}Z)&=&(X^LY^LZ^R(F)-X^RZ^RY^L(F))|_M=
(X^L-X^R)(Y^LZ^RF)|_M \\
&=&{\za}(X)g(Y,Z)\,,
\eeas
that shows that $ \nabla_X^{F^\ast}=(\nabla_X^F)^\ast$.
Note that $\nabla^F$ and  $ \nabla^{F^\ast} $ are indeed torsion-free,
$$ g(\nabla_X^FY,Z)-g(\nabla_Y^FX,Z)=X^LY^LZ^RF|_M-Y^LX^LZ^RF|_M=[X,Y]^LZ^R|F_M=g([X,Y],Z)\,.$$	
\end{proof}
\begin{corollary}
	The three tensor defined by
$$ T^F(X,Y,Z)=g(\nabla_X^FY-\nabla_X^{F^\ast}Y,Z) $$ is totally symmetric, $T\in\bigodot^3E^\ast$. We have $$g\left(\nabla_X^FY,Z \right)=g\left(\nabla_X^gY,Z \right)+\frac{1}{2}T(X,Y,Z)$$ and $$g\left(\nabla_X^{F^\ast}Y,Z \right)=g\left(\nabla_X^gY,Z \right)-\frac{1}{2}T(X,Y,Z)\,.$$ In particular, $T=0$ if $F=F^\ast$.
\end{corollary}
\begin{proof}
We can extend coordinates $x^a$ on $M$ to coordinates $(x^a,y^i)$ on $\mathcal{G}$ such that $ \mathrm{inv}(y^i)=-y^i $. Then,
$$ F=\frac{1}{2} g_{ij}(x)y^iy^j+o(|y|^2)\,\  \text{and}\,\  F^\ast=\frac{1}{2} g_{ij}(x)y^iy^j+o(|y|^2)\,,$$
so that $F-F^\ast$ has vanishing second jets at points of $M$. Hence, according to Lemma~\ref{L1},
$\tilde{X}\tilde{Y}\tilde{Z}(F-F^\ast)|_M$ does not depend on $\tilde{X}$, $ \tilde{Y} $, $ \tilde{Z} $ and the order of them, if they represent fixed vectors in $ \nu(\mathcal{G},M) $. Thus
$$ F=\frac{1}{2} g_{ij}(x)y^iy^j+\frac{1}{6}h_{ijk}(x)y^iy^jy^k+o(|y|^3)$$ and $$ F^\ast=\frac{1}{2} g_{ij}(x)y^iy^j-\frac{1}{6}h_{ijk}(x)y^iy^jy^k+o(|y|^3)\,,$$ where $h_{ijk}$ are coefficients of $T$.
\end{proof}
\begin{example}\label{ex2}
	Let $F:U(n)\rightarrow\mathbb{R}$ be a function on the unitary group given by $$F(U)=\tr(I-U)(I-U)^\dagger\,.$$ Similarly as in Example~\ref{ex1} we get for $X,Y\in \mathrm{Lie}(U(n))=\mathfrak{u}(n)$, $g^F(X,Y)=-2\tr(XY)$. Here, $T=0$  since $F=F^\ast$. Hence, $\nabla^F$ will be the Levi-Civita connection for $g^F(X,Y)=-2\tr(XY)$. We have
$g\left(\nabla_X^FY,Z \right) =X^LY^LZ^RF_M$. Consequently,
\beas Z^RF(U)&=&\frac{d}{ds}_{|{s=0}}\, F(\mathrm{e^{sZ}U})\\&=&\frac{d}{ds}_{|{s=0}}\, \tr\left(2I-\left( \mathrm{e}^{sZ}U+U^\dagger\mathrm{e}^{sZ}\right) \right)\\&=&\tr(-ZU+U^\dagger Z^\dagger)\,.
\eeas
Similarly, $ Y^LZ^RF(U)=-\tr\left(ZUY+Y^\dagger U^\dagger Z^\dagger\right) $ and
	$$X^LY^LZ^RF(I)=-\tr\left(ZXY+Y^\dagger X^\dagger Z^\dagger\right)=-\tr\left( \left[ X,Y\right] Z\right) =g\left( \frac{1}{2}\left[ X,Y\right] ,Z\right)\,.$$
	Hence, $\nabla_X^FY=\frac{1}{2}\left[ X,Y\right]$.
	
\end{example}

We will call such a $(g,\nabla, \nabla^\ast)$-structure on a Lie algebroid $E=\mathrm{Lie}(\mathcal{G})$, with $ \nabla$ and $\nabla^\ast$ being dual $E$-connections with respect to the metric $g$, a \emph{dualistic structure}.

\section{Statistical vector bundles}
If the contrast function $F$ is not metric, the Levi-Civita connection does not have a clear sense, but both symmetric tensors $g^F$ and $T^F$ are still properly defined. However, they do not depend on the Lie algebroid structure on $E$ as in the case of connections.

Since there is a tubular neighhbourhood of $M$ in $\cG$ which is diffeomorphic with a neighbourhood of $M$ in $E$ and such that $\mathsf{inv}$ acts as the multiplication by $-1$, we can always view contrast functions as defined on a neighbourhood of the zero-section in $E$. In the classical situation, we replace two-point function with a function on the tangent bundle.  Now, we define symmetric tensors $g^F\in\Sec(\bigodot^2E^*)$ and $T^F\in\Sec(\bigodot^3E^*)$ by
$$g^F(X,Y)=\tilde{X}\tilde{Y}F|_M\,,\ T^F(X,Y,Z)=\tilde{X}\tilde{Y}\tilde{Z}(F-F\circ\mathsf{inv})|_M\,.$$
Here, $\tilde{X}\,,\tilde{Y}\,,\tilde{Z}$ are any vector fields on $E$ whose values at points of $M$ coincide with the
values of vertical lifts $X^{\sv},Y^\sv,Z^\sv$ of $X,Y,Z$, respectively. Actually, we can take just the vertical lifts:
\be g^F(X,Y)={X}^\sv{Y}^\sv F|_M\,,\ T^F(X,Y,Z)={X}^\sv{Y}^\sv{Z}^\sv(F-F^*)|_M\,,\ee
where $F^*=F\circ\mathsf{inv}$.
For instance, in affine coordinates $(x^a,y^i)$ on $E$ one has ${X}^\sv=y^i(x)\pa_{y^i}$ for $X(x)=(y^i(x))$.

Thus, we can see contrast functions as defined on $E$, so that we have the associated triple $(E, g^F, T^F)$ consisting of the vector bundle $E$, and symmetric covariant 2- and 3-tensors $g^T$  and $T^F$, respectively.

Such a triple we will call a \emph{statistical vector bundle}. This is with the analogy to the terminology of Lauritzen in \cite{Amari1987} (cf. also \cite{Ciaglia2017}), where statistical vector bundles for $E=\sT M$ and $g^F$ being metric are called \emph{statistical manifolds}. If $g$ is metric, we will speak about \emph{metric statistical vector bundles}, and if $g$ is additionally positively defined -- about \emph{Riemannian statistical vector bundles}.

\begin{proposition} Any statistical vector bundle structure on $E$ comes from a globally defined contrast function on $E$. In particular, if $E=\mathrm{ Lie}(\mathcal{G})$ then any torsion-free dualistic structure on $E=\mathrm{ Lie}(\mathcal{G})$ comes from a global contrast function on $\mathcal{G}$.
\end{proposition}
\begin{proof}
Take a statistical vector bundle $(E,g,T)$. Let $ \{U_\alpha \} $ be a locally finite covering of $M$ by coordinate systems $ \left(x^a_\alpha \right)  $ and $ \tilde{U}_\alpha=\pi^{-1}(U_\alpha) $, where $\pi:U\rightarrow M$ is the projection.
In affine coordinates  $ \left(x^a_\alpha,y^i_\alpha \right)  $ on $\tilde{U}_\za$, the tensors $g,T$ take the forms
$$g(x)=\frac{1}{2} g_{ij}(x)\xd y^i_\alpha \xd y^j_\alpha\,,\quad\, T(x)=\frac{1}{6}h_{ijk}(x)\xd y^i_\alpha \xd y^j_\alpha \xd y^k_\alpha\,.$$
Let $\varphi_\alpha\ge 0 $, $\sup\varphi_\alpha\in\tilde{U}_\alpha  $, $ \sum\varphi_\alpha=1 $ on a neighborhood of M, and let us put now $F=\sum\varphi_\alpha F^\alpha$, where $F^\alpha:U^\alpha\rightarrow\mathbb{R}$ and $$F^\alpha(x,y)=\frac{1}{2} g_{ij}(x)y^i_\alpha y^j_\alpha +\frac{1}{6}h_{ijk}(x)y^i_\alpha y^j_\alpha y^k_\alpha\,.$$
 Here, on $U_\alpha$, $g(\partial y^i_\alpha, \partial y^j_\alpha) = g_{ij}(x)$, and $ T(\partial y^i_\alpha,\partial y^j_\alpha,\partial y^k_\alpha)=h_{ijk}(x) $.
For $X,Y$ being sections of $E$ we clearly have $$X^\sv Y^\sv F(x)=\sum_{x\in U_\alpha}\varphi_\alpha(x) X^\sv Y^\sv F^\alpha(x) = \sum_{x\in U_\alpha}\varphi_\alpha(x)g(X,Y)=g(X,Y)\,.$$
Similarly, $X^LY^LZ^L\left(F-F^\ast \right)=T(X,Y,Z)$.

In the case of a torsion free dualistic structure on a Lie groupoid $\mathcal{G}$, we carry over the whole structure on
$E=\mathrm{Lie}(\mathcal{G})$ using a tubular neighborhood $J:E\supset U\rightarrow\mathcal{G}$  of $M$ in $\mathcal{G}$.
\end{proof}
The above proposition is a generalization of the result by  Matumoto \cite{Matumoto1993} telling that any statistical manifold has a two-point contrast function.

\section{Higher contrast functions}

Of course, according to Lemma \ref{L1}, starting with functions $F$ with higher order jets vanishing on a submanifold $M\subset N$ (higher contrast functions), we get higher order symmetric tensors on the normal vector bundle $\zn(N,M)$.

A smooth function on a Lie groupoid $\mathcal{G}\rightrightarrows M$ (or a Lie algebroid $E=\mathrm{Lie}(\mathcal{G})$, which gives analogous results) we call a \emph{contrast function of degree} $k$ if the $k$th jet of $F$ vanishes on $M$, $\mathsf{j}^kF|_M=0$. Then we can define symmetric tensors
$$g_{k+1}^F\in\bigodot^{k+1}E^*=\bigodot^{k+1}\nu^*(\cG,M)\,,\quad T_{k+2}^F\in\bigodot^{k+2}E^*=\bigodot^{k+2}\nu^*(\cG,M)$$
by
\bea g_{k+1}^F(X_1,\dots,X_{k+1})&=&X_1^L\cdots X_{k+1}^L(F)|_M\,,\\
T_{k+1}^F(X_1,\dots,X_{k+2})&=&X_1^L\cdots X_{k+2}^L\left(F+(-1)^kF\circ\mathsf{inv}\right)|_M\,,
\eea
or, on the Lie algebroid,
\bea g_{k+1}^F(X_1,\dots,X_{k+1})&=&X_1^\sv\cdots X_{k+1}^\sv(F)|_M\,,\\
T_{k+1}^F(X_1,\dots,X_{k+2})&=&X_1^\sv\cdots X_{k+2}^\sv\left(F+(-1)^kF\circ\mathsf{inv}\right)|_M\,.
\eea
The existence of $T_{k+1}^F$ easily follows from the fact that $F+(-1)^kF\circ\mathsf{inv}$ is a contrast function of degree $k+1$ if $F$ is of degree $k$. Such a triple $(E,g_{k+1}, T_{k+2})$ we will call a \emph{statistical vector bundle of degree $k$}. As before, any such structure is generated by some contrast function of degree $k$.
In the case $\cG=M\ti M$ we can speak about \emph{statistical manifold of degree $k$}.

Note that the tensor $g^F$ could be understood as `higher metric' if it is non-degenerate in the sense that
$$g_{k+1}^F(X,\cdot,\dots,\cdot)=0\ \Leftrightarrow X=0\,.$$
In other words a `higher metric' would be the symmetric analogue of a multisymplectic form.

Starting with a contrast function $F$ of degree $k$ on a Lie groupoid $\cG\rightrightarrows M$ with the Lie algebroid $E=\mathrm{Lie}(\cG)$, and with the analogy to $H^F(X,Y,Z)=g^F(\nabla_XY,Z)$ (cf. (\ref{b})), we can define a $(k+2)$-linear map
$$H^F_{k+2}:\Sec(E)\ti\cdots\ti\Sec(E)\to C^\infty(M)$$
by
\be H^F_{k+2}(X_1,\dots,X_{k+2})=X_1^L\cdots X_{k+1}^LX^R_{k+2}F|_M\,.\ee
It is easy to see that, for $f\in C^\infty(M)$,
$$H^F_{k+2}(fX_1,\dots,X_{k+2})=fH^F_{k+2}(X_1,\dots,X_{k+2})=\ H^F_{k+2}(X_1,\dots,fX_{k+2})$$
and
$$H^F_{k+2}(X_1,\dots,fX_{l+1},\dots, X_{k+2})=\left({\za}(X_1)+\cdots{\za}(X_l)\right)(f)H^F_{k+2}(X_1,\dots,X_{k+2})\,,$$
so that $H^F_{k+2}$ is a multi-differential operator of the first order.
This looks like a definition of a sort of `higher connection' which is a local invariant (concomitant) of $F$.
Moreover,
\be H^g_{k+2}(X_1,\dots,X_{k+2})=\frac{1}{2}X_1^L\cdots X_{k+1}^LX^R_{k+2}(F-(-1)^kF^*)|_M\ee
depends on $g=g^F$ only. This is a local invariant of the `higher metric' $g^F$ and can be viewed like a definition of a `higher Levi-Civita connection'. The closer study of the concomitant $H^g_{k+2}$ we postpone, however, to a separate paper.

\section{Reductions: contrast functions on $G$ groupoids}
To describe an example of a reduction of a contrast function, consider a principal action of a Lie group $G$ on a Lie groupoid $\mathcal{G}\rightrightarrows M$. Such a structure is called in \cite{Bruce2015} a $G$-\emph{groupoid} if $G$ acts on $\mathcal{G}$ by groupoid isomorphisms. The concept of a $G$-groupoid is essentially of double nature: a $G$-groupoid is a $G$-principal bundle object in the category of Lie groupoids. Similarly, \emph{$G$-algebroids} are Lie algebroids with a principal action of $G$ by Lie algebroid isomorphisms. It is easy to see that the Lie algebroid $E=\mathrm{Lie}(\cG)$ of a $G$-groupoid is canonically a $G$-algebroid.

Let us recall that, for $\mathcal{G}_{i} \rightrightarrows  M_{i}$ $(i=1,2)$ being a pair of Lie groupoids, a \emph{Lie groupoid morphism} is a pair of maps $(\Phi, \phi)$ such that the following diagram is commutative\smallskip

\begin{tabular}{p{5cm} p{10cm}}
\begin{xy}
(0,20)*+{\mathcal{G}_{1}}="a"; (20,20)*+{\mathcal{G}_{2}}="b";%
(0,0)*+{M_{1}}="c"; (20,0)*+{M_{2}}="d";%
{\ar "a";"b"}?*!/_2mm/{\Phi};
{\ar@<1.ex>"a";"c"} ;
{\ar@<-1.ex> "a";"c"} ?*!/^3mm/{\mathsf{s}_{1}} ?*!/_6mm/{\mathsf{t}_{1}};
{\ar@<1.ex>"b";"d"};%
{\ar@<-1.ex> "b";"d"}?*!/^3mm/{\mathsf{s}_{2}} ?*!/_6mm/{\mathsf{t}_{2}};  %
{\ar "c";"d"}?*!/^3mm/{\phi};
\end{xy}
&
\vspace{-60pt}
\noindent in the sense that\smallskip

\noindent $\mathsf{s}_{2}\circ \Phi = \phi \circ \mathsf{s}_{1},  \hspace{20pt}\textnormal{and} \hspace{20pt} \mathsf{t}_{2}\circ \Phi = \phi \circ \mathsf{t}_{1} $,
\end{tabular}
\smallskip

\noindent subject to the further condition that $\Phi$ respects the (partial) multiplication; if $a,b \in \mathcal{G}_{1}$ are composable, then  $ \Phi(ab) = \Phi(a)\Phi(b)$. It then follows that for $x \in M_{1}$  we have $\Phi(\id_{x}) = \id_{\phi(x)}$ and  $\Phi(a^{-1}) = \Phi(a)^{-1}$. Like in the classical Lie Theory, morphisms of Lie groupoids induce morphisms of the corresponding Lie algebroids (see \cite{Mackenzie2005}).

\medskip
For a $G$ groupoid $\mathcal{G}\rightrightarrows M$,

\begin{enumerate}
\item The action of $G$ on $\mathcal{G}$ commutes with the source and target maps, thus projects onto a $G$-action on the manifold $M$. Moreover, $M$ as an immersed submanifold of $\cG$ is invariant with respect to the $G$-action, and the projected and restricted actions coincide.
\item As the action of $G$ on $\mathcal{G}$ is principal, it is also principal on the immersed submanifold $M$, so $M$ inherits a structure of a principal $G$-bundle.   It is important to note that $M$ is G-invariant. In particular, the quotient manifold $M_0=M/G$ exists.
 \item The reduced manifold $\cG/G=\mathcal{G}_{0}$ is a Lie groupoid $\cG/G=\mathcal{G}_{0} \rightrightarrows M/G=M_{0}$, with the set of units $M_0$, defined by the following structure:
\end{enumerate}
\begin{tabular}{p{5cm} p{5cm}}
\begin{xy}
(0,20)*+{\mathcal{G}}="a"; (20,20)*+{\mathcal{G}_{0}}="b";%
(0,0)*+{M}="c"; (20,0)*+{M_{0}}="d";%
{\ar "a";"b"}?*!/_2mm/{\pi};
{\ar@<1.ex>"a";"c"} ;
{\ar@<-1.ex> "a";"c"} ?*!/^3mm/{\mathsf{s}} ?*!/_6mm/{\mathsf{t}};
{\ar@<1.ex>"b";"d"};%
{\ar@<-1.ex> "b";"d"}?*!/^3mm/{{\zs}} ?*!/_6mm/{{\zt}};  %
{\ar "c";"d"}?*!/^3mm/{\rmp};
\end{xy}
&\vspace{-80pt}
{\begin{align*}
 &{\zs} \circ \pi = \rmp \circ \mathsf{s}\,, \\
 &{\zt} \circ \pi = \rmp \circ \mathsf{t}\,,\\
 & \id_{\rmp(x)}=\zp(\id_x)\quad\text{for all}\quad x\in M\,,\\
 & \pi(a)^{-1}=\pi(a^{-1})\quad\text{for all}\quad a\in\cG\,,\\
 &\pi(a a') = \pi(a) \pi(a')\quad\text{for all}\quad(a,a') \in \mathcal{G}^{(2)}\,,
\end{align*}}
\end{tabular}

\noindent where $\pi:\cG\to\cG_0$ is the canonical projection.
In fact, the above constructions imply, tautologically, that $(\pi, \rmp) : \mathcal{G} \rightrightarrows M \rightarrow \mathcal{G}_{0}\rightrightarrows M_{0}$ is a morphism of Lie groupoids with the above structures. The fundamental fact in the Lie theory of groupoids says that any morphism $\Phi$ of Lie groupoids induces a morphism $\mathrm{Lie}(\Phi)$ of the corresponding Lie algebroids. It is derived from the map $\sT\Phi|_M$. In our case,
$$(\Pi,\rmp)=\mathrm{Lie}(\pi,\rmp):E=\mathrm{Lie}(\cG)\to E_0=\mathrm{Lie}(\cG_0)$$
is covering the map $\rmp:M\to M_0=M/G$. It is easy to see that $E_0=E/G$. This Lie algebroid morphism defines in turn the pul-backs of symmetric forms
\be \Pi^*:\Sec\left(\bigodot^kE^*_0\right)\to\Sec\left(\bigodot^kE^*\right)\,, \ \Pi^*(\zw)_x(X_1,\cdots,X_k)=
\zw(\sT_x\Pi(X_1),\cdots, \sT_x\Pi(X_1))\,.
\ee
Note that a morphisms of vector bundles do not induce, in general, maps of the corresponding sections. This makes the definition of a Lie algebroid morphism non-trivial.

\begin{example}\label{e1}
Let $P\to M$ be a $G$-principal bundle with the right action $$P\ti G\to P\,,\quad(a,\zg)\mapsto a\zg\,.$$
Then, the pair Lie groupoid $P\ti P$ is a $G$-groupoid with respect to the action $(a,b)\zg\mapsto (a\zg,b\zg)$ and we have the corresponding morphism of Lie groupoids
$$\begin{xy}
(0,20)*+{P\ti P}="a"; (40,20)*+{(P\ti P)/G}="b";%
(0,0)*+{P}="c"; (40,0)*+{P/G=M}="d";%
{\ar "a";"b"}?*!/_2mm/{\pi};
{\ar@<1.ex>"a";"c"} ;
{\ar@<-1.ex> "a";"c"} ?*!/^3mm/{\mathsf{s}} ?*!/_6mm/{\mathsf{t}};
{\ar@<1.ex>"b";"d"};%
{\ar@<-1.ex> "b";"d"}?*!/^3mm/{{\zs}} ?*!/_6mm/{{\zt}};  %
{\ar "c";"d"}?*!/^3mm/{\rmp};
\end{xy}
$$
The Lie groupoid $(P\ti P)/G\rightrightarrows M$ is called the \emph{Atiyah groupoid} of the principal bundle $P\to M$
The corresponding Lie algebroid $$E_0=\mathrm{Lie}((P\ti P)/G)= \sT P/G$$ is the Atiyah algebroid with the \emph{Atiyah exact sequence} of Lie algebroid morphisms
$$0\to K\to \sT P/G\to \sT M\to 0\,,$$ where $P/G\to \sT M$ is the (surjective in this case) anchor map and $K$ is its kernel, a bundle of Lie algebras isomorphic to $\mathfrak{g}=\mathrm{Lie}(G)$. The map $\Pi^*$ identifies sections of $\sT^*P/G$ with $G$-invariant $1$-forms on $P$.
\end{example}
The structure of a $G$-groupoid is described in \cite{Bruce2015}. For simplicity, we present the result for trivial $G$-structures (all $G$-groupoids are locally trivial).
\begin{theorem}\label{trivialsplit} For any $G$-groupoid structure on the trivial $G$-bundle $\cG=\cG_0\ti G$ there is a Lie groupoid structure on $\cG_0$
with the source and target maps $\zs,\zt:\cG_0\to M_0$ and a groupoid morphism $\rmb:\cG_0\to G$ such that the source map $s$, the target map $t$ and the partial multiplication in $\cG$ read
\be\label{Gg}s(y_0,\zg)=(\zs(y_0), \zg)\,,\quad t(y_0,\zg)=(\zt(y_0),\rmb(y_0)\zg)\,,\quad
(y_0,\zg)(y'_0,\zg')=(y_0y'_0,\zg')\,.
\ee
\end{theorem}
\noindent Let us see what is the structure of $E=\mathrm{Lie}(\cG)=\mathrm{Lie}(\cG_0)\ti G$.

First, according to the decomposition $E=E_0\ti G\to M_0\ti G$, we can view sections $E$ as $\zg$-dependent sections of $E_0$ with the identification of sections $X$ of $E_0$ as $\zg$-independent sections ${\bX}$ of $E$, i.e. $\bX(x,\zg)=X(x)$. Since the $\zg$-independent sections generate the module $\Sec(E)$ over $C^\infty(M_0\ti G)$, the bracket in $\Sec(E)$ is completely determined by the bracket of the $\zg$-independent sections and the anchor map. Hence, the left (resp., right) invariant vector fields on $\cG$ are spanned by $f^L\bX^L$ (resp., $f^R\bX^R$), where $f\in C^\infty(M_0\ti G)$.

Now, using the decomposition $\sT\cG=\sT\cG_0\ti \sT G$, by straightforward calculations we obtain:
\begin{proposition}\label{p1}
For $X\in\Sec(E_0)$, $f\in C^\infty(M_0\ti G)$, we have in the groupoid $\cG=\cG_0\ti M$:
\beas\mathrm{inv}(a,\zg)&=&(a^{-1},b(a)\zg)\,,\ \bX^R(a,\zg)=X^R(a)\,,\ \bX^L(a,\zg)=X^L(a)+\sT_{a^{-1}}b(X^R(a^{-1}))\zg\,,\\
\za_\cG(\bX)&=&\za_{\cG_0}(X)+\mathrm{Lie}(b)(X)^L_G\,,\ f^L(a,\zg)=f(\zs(a),\zg)\,, \ f^R(a,\zg)=f(\zt(a), b(a)\zg)\,.
\eeas
Here, of course, for $Y\in\sT_eG$, the symbol $Y\zg$ denotes the right translation of the vector $Y$ by $\zg\in G$.
\end{proposition}
Note that the above proposition implies easily that $[\bX,\bX'](x,g)=[X,X'](x)$, so $\zg$-independent sections of $E$ commute as they representatives in $\Sec(E_0)$, and the knowledge of the anchor completely determines the Lie bracket in $\Sec(E)$.

\begin{proposition}
Assume that $F:\cG\to\R$ is a contrast function on a $G$ groupoid $\mathcal{G}\rightrightarrows M$ and $\pi:\cG\to\cG_0=\cG/G$ is the corresponding groupoid morphism.
If $F:\cG\to\R$ is $G$-invariant contrast function, then $F$ induces a contrast function $\hat{F}$ on $\cG_0=\cG/G$ and
$$g^F=\mathrm{Lie}(\pi)^*(g^{\hat{F}})\quad\text{and}\quad T^F=\mathrm{Lie}(\pi)^*(T^{\hat{F}})\,.$$
In other words,
\be\label{red}g^F(\bX,\bY)(x,\zg)=g^{\hat{F}}(X,Y)(x)\ \text{and}\ T^F(\bX,\bY,\bZ)(x,\zg)=T^{\hat{F}}(X,Y,Z)(x)\,.\ee
Moreover, if $F$ is metric, then $\hat{F}$ is metric and
\be\label{red1}\nabla^F_\bX\bY=\mathbf{V}\,, \ \text{where}\  V=\nabla^{\hat{F}}_XY\,.
\ee
\end{proposition}
\begin{proof}
According to Proposition \ref{p1}, $\bX^L$ and $\bX^R$ differ from $X^L$ and $Y^L$ (viewed as vector fields on $M_0\ti G$) by vector fields tangent to orbits of $G$. As $F$ is $G$-invariant and
$$F^*(a,\zg)=F\circ\mathrm{inv}_\cG (a,\zg)=F(a^{-1},b(a)\zg)=F(a^{-1},e)=\hat{F}\circ\mathrm{inv}_{\cG_0}(a)=\hat{F}^*(a)\,,$$
the first jets of $F$ along $M=M_0\ti G$ and $\hat{F}$ along $M_0$ are trivial, and
$$\bX^L\bY^LF(x,\zg)=X^LY^L\hat{F}(x)\,,\ \bX^L\bY^L\bZ^L(F-F^*)(x,\zg)=X^LY^LZ^L(\hat{F}-\hat{F}^*)(x)\,,$$
so (\ref{red}) follows. Similarly,
$$g^F(\nabla^F_\bX\bY,\bZ)(x,\zg)=\bX^L\bY^L\bZ^RF(x,\zg)=X^LY^LZ^R\hat{F}(x)=g^{\hat{F}}(\nabla^{\hat{F}}_XY,Z)\,,
$$
whence (\ref{red1}).
\begin{example}
Let $P$ be the $\SO(n)$-principal bundle of oriented orthonormal frames on the sphere $S^n$ canonically embedded in $\R^{n+1}$ as the unit sphere. In other words, $P$ is the orhonormal frame bundle of $\sT S^n$ with the canonical Riemannian metric. We will view $P$ as the set of pairs $(x,r)$, where $x\in S^n$ and $r:\sT_xS^n\to \R^n$ is an isometry respecting the orientations. Note that elements $(x,r)\in P$ can be identified with orthonormal frames in $\R^{n+1}$ in the obvious way. This implies that $P$ is simultaneously a homogeneous space of the group $\SO(n+1)$ acting freely on $P$ (transitive $\SO(n+1)$-principal bundle). The group $\SO(n)$ acts as a subgroup of $\SO(n+1)$ with respect to the embedding $\SO(n)\ni\zg\mapsto \bar{\zg}\in\SO(n+1)$,
$$ \bar{\zg}=
  \left[ {\begin{array}{cc}
   1 & 0 \\
   0 & \zg \\
  \end{array} } \right]\,.
  $$

Consider the pair groupoid $\cG=P\ti P\rightrightarrows P$. To every pair
$(p,p')\in P\ti P$ we can associate a matrix $A({p,p'})$ which is the unique matrix in $\SO(n+1)$ which maps the oriented orthonormal frame $p=(x,r)$ of $\R^{n+1}$ onto $p'=(x',r')$. This defines the map $A:P\ti P\to SO(n+1)$ and, in turn, a diffeomorphism
$$\Phi:P\ti P\to P\ti SO(n+1)\,, \ \Phi(p,p')=(p,A(p,p'))$$
which identifies the diagonal $\zD\simeq P$ in $P\ti P$ with $P\ti\{ I\}\subset P\ti\SO(n+1)$.
The normal bundle $\zn(P\ti P,\zD)$ of $\zD$ in $P\ti P$, thus $\sT P$, can be therefore identified with $P\ti\so(n+1)$. With this identification, element $(p,X)\in P\ti\so(n+1)$ corresponds to the vector $\tilde{X}(p)\in\sT_pP$ which is tangent to the curve $t\mapsto p\exp(tX)\in P$ at $t=0$. Hence, $\tilde{X}$ is the fundamental vector field of the $\SO(n+1)$-action on $P$, corresponding to $X\in\so(n+1)$.

Define now the two-point function $F:P\ti P\to\R$
$$F(p,p')=\frac{1}{4}\tr\left(2I-A(p,p')-A(p,p')^t\right)\,.$$
This function is a contrast function: it vanishes on the diagonal
and
$$\frac{d}{dt}_{|t=0}F(p,p\exp(tX))=\frac{1}{4}\frac{d}{dt}_{|t=0}\tr\left(2I-\exp(tX)-\exp(-tX)\right)=\frac{1}{4}\tr(-X+X)=0\,.$$
We obtain the metric $g^F$ on the normal bundle $\sT P\simeq P\ti\so(n+1)$ similarly like in Example \ref{ex1},
\be\label{met}g^F(p)(\tilde{X},\tilde{Y})=\frac{1}{2}\tr(XY^t)=-\frac{1}{2}\tr(XY)\,.\ee
The function $F$ is invariant with respect to the inversion $\mathrm{inv}(p,A)=(p,A^{-1})$, so $T^F=0$.
Like in Example \ref{ex2} we obtain the Levi-Civita connection in the form
$$\nabla^F_{\tilde{X}}\tilde{Y}=\frac{1}{2}\widetilde{[X,Y]}\,,$$
where the bracket is that in the Lie algebra $\so(n+1)$.

Let us observe now (cf. Example \ref{e1}) that $P\ti P$ is a $SO(n)$-groupoid with respect to the obvious action
$$((x,r),(x',r'))\zg=((x,r\zg),(x',r'\zg))=\left((x,r)\bar{\zg},(x',r')\bar{\zg}\right)\,.$$
This action, in the identification $P\ti P\simeq P\ti\SO(n+1)$, looks like
$$(p,A)\zg=(p\zg,\bar{\zg}^{-1}A\bar{\zg})\,,$$
hence $\SO(n)$ acts on $\sT P\simeq P\ti\so(n+1)$ \emph{via}
$$(p,X)\zg=\left(p\zg, \Ad_{\bar{\zg}}^{-1}(X)\right)\,,$$
so the vector field $(p,X(p))$ on $P$ is $\SO(n)$-invariant if $X(p\zg)=\Ad_{\bar{\zg}}^{-1}(X)$.
Note the canonical decomposition
$$\sT P=P\ti\so(n)\ti\mathfrak{k}\,,$$
where $\mathfrak{k}$ is the orthogonal complement of $\so(n)$ in $\so(n+1)$ with respect to the trace scalar product.
Note that this decomposition is ${\SO(n)}$-independent.

The Lie algebroid $E_0=\mathrm{Lie}(\cG_0)\to S^n$ is $\sT P/\SO(n)\to S^n$.
Its sections are identified with $\SO(n)$-invariant vector fields on $P$.
It is clear that $F$ is $\SO(n)$-invariant,
$$\tr\left(2I-\bar{\zg}^{-1}A\bar{\zg}-\bar{\zg}^{-1}A^t\bar{\zg}\right)=\tr\left(\bar{\zg}^{-1}(2I-A-A^t)\bar{\zg}\right)=\tr(2I-A-A^t)\,,$$ so it defines a metric contrast function
$\hat{F}$ on the Atiyah groupoid
$$\cG_0=(P\ti P)/SO(n)\rightrightarrows S^n\,.$$

We can simplify the picture choosing one point of $P$, say $(x_0,r_0)=(e_1,\dots,e_{n+1})$ to identify $P$ with $\SO(n+1)$ and $\tilde{X}$ with right-invariant vector fields $X^r$ on $\SO(n+1)$. This time, however, the left-invariant vector fields $X^l$ represent $\SO(n)$-invariant vector fields on $P$, so sections of $E_0$. Moreover, the bracket in the Lie algebroid $E_0$ on sections $X^l$ agrees with the Lie bracket in $\so(n+1)$,
$$[X,Y]^l=[X^l,Y^l]\,.$$
Actually, the invariance of $g^F$ is with respect to the $\Ad$-action of the whole $\SO(n+1)$, so the metric (\ref{met})
induces a Riemannian metric on $SO(n+1)$ which is simultaneously left- and right-invariant.
There is  a canonical mapping $\tilde{\za}:\sT P\to \sT S^n$ , obtained from the submersion
$$\SO(n+1)\to\SO(n+1)/SO(n)\simeq S^n\,,$$
which induces the anchor map $\za:E_0\to\sT S^n$. The left invariant vector fields $X^l$ on $P=\SO(n+1)$ generate now the module $\Sec(E_0)$ as a module over $C^\infty(S^n)$-the $\SO(n)$-invariant functions on $P$. The anchor $\za(X^l)$ is the corresponding fundamental vector field of the canonical action of $\SO(n+1)$ on $\R^{n+1}$. Thus $X^l$ projects under $\za$ to a $\SO(n+1)$-invariant vector field $\za(X)$ on $S^n$. The kernel of this projection is generated by
$X^l$, where $X\in\SO(n)\subset\SO(n+1)$, so the anchor map identifies $\mathfrak{k}$ with $\sT_{e_1}S^n$.

Tu sum up: We can identify $\cG_0=(P\ti P)/SO(n)$ with $$\cG_0=(\SO(n+1)\ti\SO(n+1))/SO(n)=S^n\ti\SO(n+1)\,,$$ so that the reduced contrast function is
\be\label{recont}
\hat{F}(x,A)=\frac{1}{4}\tr(2I-A-A^t)\,,\ (x,A)\in S^n\ti\SO(n+1)\,.
\ee
The Lie algebroid $E_0=\mathrm{Lie}(\cG_0)\to S^n$ is $\sT P/\SO(n)\to S^n$ identified with the  normal bundle, i.e. $E_0=S^n\ti\so(n+1)$. Any constant section
$X(p)=X\in\so(n+1)$ represents a $\SO(n)$-invariant vector field $X^l$ on $P$ (left-invariant vector field on $\SO(n+1)$) and projects to a $\SO(n+1)$-invariant vector field on $S^n=P/\SO(n)$.
The reduced contrast function (\ref{recont}) induces on $E_0$ a metric $g^{\hat{F}}$ by
$$g^{\hat{F}}(X^l,Y^l)=-\frac{1}{2}\tr(XY)\,.$$ Moreover, $T^{\hat{F}}=0$ and the Lie algebroid Levi-Civita connection for $g^{\hat{F}}$ which satisfies
$$\nabla^{\hat{F}}_{X^l}Y^l=-\frac{1}{2}[X,Y]\,.$$
The connection is clearly torsionless. The full form of $\nabla^{\hat{F}}$ involves of course the anchor map. In particular, for $X,Y\in\so(n)$, $f,g\in C^\infty(S^n)$, we have
$$\nabla^{\hat{F}}_{fX^l}(gY^l)=\frac{1}{2}fg[X,Y]^l\,,$$
since the anchors $\za(X))==\za(Y)=0$ are trivial.
Due to invariance, the metric and the connection project to the sphere $S^n$. Using the base
$$v_i=\zd^1_i-\zd^i_1\in\mathfrak{k}\subset\so(n+1)\,,\ i=2,\dots,n+1$$ in $\mathfrak{k}$ corresponding to vectors
$e_2,\dots,e_{n+1}\in\sT_{e_1}S^n$, we see that
$$g^{\hat{F}}(v_i,v_j)=-\frac{1}{2}\tr(v_i,v_j)=\zd^i_j\,,$$
that shows that the metric induced by $g^F$ on the sphere is the standard Riemannian metric.
Hence, $\nabla^{\hat{F}}_{v_i^l}(v_j^l)=0$ and
$$\nabla^{\hat{F}}_{fv_i^l}(gv_j^l)=f\za(v_i^l)(g)v_j^l\,.$$
\end{example}
\end{proof}

\section{Infinite dimensions}\label{sec:infdim}
Our coordinate-free approach to stochastic manifolds has an additional advantage: it can be applied practically without changes in infinite-dimensional, say Banach manifold, frameworks. The differential calculus on Banach manifolds, in particular Banach-Lie groupoids, produces forms $g^F(x)$ as elements of $V^*\ti V^*$ for some Banach spaces $V$. This time, however, the non-degeneracy is a more delicate problem. This is due to the fact that Banach manifolds are generally not reflexive, the more not self-dual. In a weaker version, for non-degeneracy of $g^F(x)$ one can assume that the map $V\ni Y\to g^F(x)(Y,\cdot)\in V^*$ is an immersion, in a strong one, that it is an
isomorphism. The latter require of course that $V$ is self-dual, $V\simeq V^*$.

The best infinite-dimensional framework is therefore that of (real or complex) Hilbert spaces. Here is a nice example.
\begin{example}
For a Hilbert space $\cH$, on $\cH^\ti\ti \cH^\ti$, where $\cH^\ti=\cH\setminus\{ 0\}$, consider the two-point function
$$F(\zf,\zc)=1-\frac{|\bk{\zf}{\zc}|^2}{||\zf||^2\cdot||\zc||^2}\,.$$
It is easy to see that $F$ is a non-negative contrast function. Indeed, calculating the derivative with respect to $\zf$, we get
$$x^LF(\zf,\zc)=\frac{d}{dt}_{|t=0}F(\zf+tx,\zc)=\frac{2\re\bk{\zf}{x}|\bk{\zf}{\zc}|^2-2\re(\bk{x}{\zc}\bk{\zc}{\zf})||\zf||^2}
{||\zf||^4\cdot||\zc||^2}\,,$$
so that $x^LF(\zf,\zf)=0$. Now,
\beas x^Ly^LF(\zf,\zf)&=&\frac{d}{dt}_{|t=0}\frac{2\re\bk{\zf+ty}{x}|\bk{\zf+ty}{\zf}|^2-2\re(\bk{\zf}{x}\bk{\zf+ty}{\zf})||\zf+ty||^2}
{||\zf+ty||^4\cdot||\zf||^2}\\
&=& \frac{2\re\bk{x}{y}||\zf||^2-2\re(\bk{x}{\zf}\bk{y}{\zf})}{||\zf||^4}=g^F(\zf)(x,y)\,.\eeas
The 2-form $g^F$  is degenerated, but if we reduce by the action of $\C^\ti\ti\C^\ti$ (the contrast function $F$
is invariant with respect to $\C^\ti\ti\C^\ti$-action
$$(\zf,\zc)(z,z')=(z\zf, z'\zc)$$
on $\cH^\ti\ti\cH^\ti$), we obtain a Riemannian metric on the Hilbert projective space $\mathbb{P}\cH=\cH^\ti/\C^\ti$. This metric reads
$$\xd g(\zf)=\frac{2||x||^2\cdot||\zf||^2-2\bk{x}{\zf}|^2}{||\zf||^4}\,\xd x^2\,,$$
i.e. it is proportional to the Fubini-Study metric on $\mathbb{P}\cH$.

\end{example}

\begin{example} The groupoid of rank-one operators.
	
\noindent We consider $S(\mathcal{H})$, the unit sphere in $\mathcal{H}$, as a $U(1)$-principal bundle over the complex projective space $\mathbb{P}(\mathcal{H})$. Using normalized vectors 
\beas
\ket{\bar{\psi}}=\frac{\ket{\psi}}{\sqrt{\bk{\psi}{\psi}}},
\eeas
we construct transition probability amplitudes $\kb{\bar{\psi}}{\bar{\phi}}$ as elements of $S(\mathcal{H})\times S(\mathcal{H})$. Equivalence classes $e^{i\varphi}\kb{\bar{\psi}}{\bar{\phi}}e^{-i\varphi} $ are projected onto:
\begin{eqnarray}
s:\kb{\bar{\psi}}{\bar{\phi}}\mapsto\kb{\bar{\phi}}{\bar{\phi}}
&=&\frac{\kb{\phi}{\phi}}{\bk{\phi}{\phi}}, \nonumber \\
t:\kb{\bar{\psi}}{\bar{\phi}}\mapsto\kb{\bar{\psi}}{\bar{\psi}}
&=&\frac{\kb{\psi}{\psi}}{\bk{\psi}{\psi}}, \nonumber		
\end{eqnarray}
hence, the Atiyah groupoid $\cG(S(\mathcal{H}))$ projects onto the complex projective space represented by rank-one operators.
\end{example}

\section{Conclusions and outlook}

In previous sections we have argued that a groupoid approach to differential geometry of information theory is a more natural setting to deal with sub-manifolds of classical probability distributions. We have also considered the reduction problem of contrast functions which will be very useful in the quantum setting, where relative entropies will be invariant under the action of the unitary group.

As a matter of fact, a coordinate free approach to deal with the differential calculus required to derive metric and dual connections out of potential or contrast functions was introduced previously \cite{ciaglia18,ciaglia19,laudato18,manko17}, however there the introduction was by “ad hoc” methods, here it is intrinsic with the notion of Lie groupoid and its associated Lie algebroid. Moreover the notion of groupoid enters also naturally within the Schwinger approach to quantum mechanics \cite{ciaglia18a,ciaglia19,ibort13}. As the example provided in Section~\ref{sec:infdim} shows, it is possible to write a contrast function in quantum mechanics. The contrast function used there arises from an Atiyah groupoid, indeed it is possible to consider $S(\mathcal{H})$, the unit sphere in the Hilbert space as a $U(1)$-principle bundle over the complex projective space, then the groupoid $P\times P/U(1)$ has a space of “objects” (units) provided by rank-one projectors which represent the pure states, the “arrows” are transition probability amplitudes.

Thus in this quantum setting, we replace probabilities with probability amplitudes and transition probabilities with transition probability amplitudes. This replacement is crucial to be able to describe quantum interference phenomena as argued by Born in his Nobel acceptance speech. We are already familiar with the interpretation of “wave functions” as  probability amplitudes. This shift from probabilities to their “complex square root” allows to introduce also in quantum mechanics the language of groupoids to deal with contrast functions.

In a forthcoming paper we shall elaborate on the groupoid setting both in the Hilbert space approach and the C*-algebra approach to quantum mechanics to deal with contrast functions considered as  generalized relative entropies.

%
%
%


\section*{Acknowledgments}
J. Grabowski acknowledges research founded by the  Polish National Science Centre grant HARMONIA under the contract number 2016/22/M/ST1/00542. M. Ku\'s acknowledges financial support of the the Polish National Science Centre grant 2017/27/B/ST2/02959. G.~Marmo acknowledges financial support from the Spanish Ministry of Economy and Competitiveness, through the Severo Ochoa Programme for Centres of Excellence in RD(SEV-2015/0554), and would like to thank the support provided by the Santander/UC3M Excellence Chair Programme 2018/2019.

%

\vskip1cm
\noindent Katarzyna Grabowska\\\emph{Faculty of Physics,
University of Warsaw,}\\
{\small ul. Pasteura 5, 02-093 Warszawa, Poland} \\{\tt konieczn@fuw.edu.pl}\\

\noindent Janusz Grabowski\\\emph{Institute of Mathematics, Polish Academy of Sciences}\\{\small ul. \'Sniadeckich 8, 00-656 Warszawa,
Poland}\\{\tt jagrab@impan.pl}
\\

\noindent Marek Ku\'s\\
\emph{Center for Theoretical Physics, Polish Academy of Sciences,} \\
{\small Aleja Lotnik{\'o}w 32/46, 02-668 Warszawa,
Poland} \\{\tt marek.kus@cft.edu.pl}
\\

\noindent Giuseppe Marmo\\
\emph{Dipartimento di Fisica ``Ettore Pancini'', Universit\`{a} ``Federico II'' di Napoli} \\
\emph{and Istituto Nazionale di Fisica Nucleare, Sezione di Napoli,} \\
{\small Complesso Universitario di Monte Sant Angelo,} \\
{\small Via Cintia, I-80126 Napoli, Italy} \\
{\tt marmo@na.infn.it}
\\

\end{document}